\documentclass{www2013-accepted}
\usepackage{algorithmic}
\usepackage{algorithm}
\usepackage{url}

\newcommand{\inner}[1]{\left\langle#1\right\rangle}

\def\DB{{\small DB}}
\def\Th{{\small T}}
\def\DM{{\small DM}}
\def\AI{{\small AI}}

\def\density{\mathrm{density}}
\def\dist{\mathrm{dist}}
\def\T{\mathcal{T}}

\def\A{\mathcal{A}}

\def\R{\mathbb{R}}

\newcommand{\norm}[1]{\left\|#1\right\|}
\newcommand{\abs}[1]{\left|#1\right|}

\def\ones{\mathbf{1}}

\def\assoc{\mathrm{assoc}}
\def\cut{\mathrm{cut}}

\def\sgn{\mathop{\rm sgn}\nolimits}

\def\vol{\mathop{\rm vol}\nolimits}

\def\argmin{\mathop{\rm arg\,min}\limits}

\def\subj{\mathop{\rm subject\,to}}

\def\ones{\mathbf{1}}

\def\cut{\mathrm{cut}}

\def\qed{\hfill$\Box$}
\def\P{\Pi}

\newtheorem{theorem}{Theorem}
\newtheorem{lemma}{Lemma}

\newtheorem{definition}{Definition}
\newtheorem{corollary}{Corollary}

\def\pen{\mathop{\rm pen}\nolimits}
\def\spvol{\mathop{\rm spvol}\nolimits}
\def\unit{\mathop{\rm unit}\nolimits}

\def\Vc{V}

\def\FORTE{\textbf{FORTE}}
\def\LPFeas{\textbf{LPfeas}}
\def\mdalk{\textbf{mdAlk}}
\def\lpfeas{\textbf{LPfeas}}

\def\Vp{V^\prime}
\def\Gp{G^\prime}







\makeatletter
\newenvironment{small-lar}[1]{\skip@=\baselineskip#1 \baselineskip=\skip@\lbegar }{\rendar \ignorespacesafterend}
\makeatother

\makeatletter
\newenvironment{small-ar}[1]{\skip@=\baselineskip#1 \baselineskip=\skip@\begar }{\endar \ignorespacesafterend}
\makeatother

\renewcommand{\theenumi}{\alph{enumi}}

\renewcommand{\theenumii}{\arabic{enumii}}

\begin{document}
%

\title{Towards Realistic Team Formation in Social Networks based on Densest Subgraphs}
%
%
%
%
%

\numberofauthors{3} 
%
\author{
%
%
\alignauthor
Syama  Rangapuram\\
       \affaddr{Max Planck Institute for Computer Science}\\
       \affaddr{Saarbr\"ucken, Germany}\\
       \email{srangapu@mpi-inf.mpg.de}
\alignauthor
Thomas B\"uhler\\
       \affaddr{Faculty of Mathematics\\ and Computer Science}\\
       \affaddr{Saarland University}\\
       \affaddr{Saarbr\"ucken, Germany}\\
       \email{tb@cs.uni-saarland.de}
\and
\alignauthor 
Matthias Hein\\
       \affaddr{Faculty of Mathematics\\ and Computer Science}\\
       \affaddr{Saarland University}\\
       \affaddr{Saarbr\"ucken, Germany}\\
       \email{hein@cs.uni-saarland.de}
}

\maketitle
\begin{abstract}
Given a task $\T$, 
 a set of experts $V$ with multiple skills and a social network $G(V, W)$ 
 reflecting the compatibility 
 among the experts, {\it team formation} is the problem of 
 identifying a team $C \subseteq V$
  that is both competent in performing the task $\T$ and compatible in working together.
Existing methods for this problem make too restrictive assumptions and thus
cannot model practical scenarios. 
 The goal of this paper is to consider the team formation problem in a realistic setting
 and  present a novel formulation based on densest subgraphs. 
 Our formulation allows modeling of many natural requirements such as
  (i) inclusion of  a designated team leader
 and/or  a group of given experts, 
 (ii) restriction of the size or more generally cost of the team
 (iii) enforcing \textit{locality} of the team, e.g., in a geographical sense or social sense, etc.
 The proposed formulation leads to a generalized version of the classical densest subgraph problem with cardinality constraints (DSP), 
 which is an NP hard problem and has many applications in social network analysis. 
 In this paper, we present a new method for (approximately) solving the generalized DSP (GDSP). 
 Our method, \FORTE, is based on solving an \textit{equivalent} continuous relaxation of GDSP.
 The solution found by our method has a quality guarantee and always satisfies the constraints of GDSP.
 Experiments show that the proposed formulation (GDSP) is useful in  modeling a broader range of
 team formation problems and that our method produces more coherent and compact teams of high quality. 
 We also show, with the help of an LP relaxation of GDSP, that
 our method gives close to optimal solutions to GDSP.
\end{abstract}

\category{H.2.8}{Database Management}{Database Applications}[Data mining]
\category{G.2.2}{Discrete Mathematics}{Graph Theory}[Graph algorithms]

\terms{Algorithms, Experimentation, Theory}

\keywords{Team formation, Social networks, Densest subgraphs}

\section{Introduction}
  Given a set of skill requirements (called task $\T$), 
 a set of experts who have expertise in one or more skill, along with 
 a social or professional network of the experts, 
 the team formation problem is to identify a competent and highly collaborative team.
 This problem in the context of a social network was first introduced by \cite{LapLiuTer09} and 
  has attracted recent interest in the data mining community \cite{KarAn11,  AnaBecCasGioLeo12, GajSar12}. 
%
A closely related and well-studied problem in operations research is the assignment problem.
Here, given a set of agents and a set of tasks, the goal is to find an agent-task assignment minimizing the cost of 
the assignment 
such that exactly one agent is assigned to a task and every task is assigned to some agent. 
This problem can be modeled as a maximum weight matching problem in a weighted bipartite graph. 
 In contrast to the assignment problem, the
 team formation problem considers the underlying social network, which for example models the previous collaborations 
 among the experts,
   while forming teams. 
 The advantage of using such a social network is that the teams that have worked together previously
 are expected to have less communication overhead and work more effectively as a team.
 
  The criteria explored in the literature so far for measuring the effectiveness of teams 
   are based on the shortest path distances, density, and  
 the cost of the minimum spanning tree of the subgraph induced by the team. 
 Here the density of a subgraph is defined as the ratio of the total weight of the edges within the subgraph over the size
  of the subgraph. Teams that are well connected have high density values.
 Methods based on minimizing diameter (largest shortest path between any two vertices) or 
 cost of the spanning tree have the main advantage that the teams they yield are always connected (provided 
 the underlying social network is connected). 
 However, diameter or spanning tree based objectives are not robust to the changes (addition/deletion of edges) in 
 the social network.
 As  demonstrated in \cite{GajSar12} using various performance measures, 
 the density based objective performs better in identifying  well connected teams.
 On the other hand, maximizing density may give a team whose subgraph is disconnected. 
 This happens especially when there are small groups of people who are highly connected with each other but 
 are sparsely connected to the rest of the graph. 
 
 Existing methods make either strong assumptions 
 on the problem that do not hold in practice or are not capable of incorporating more intuitive constraints such as
   bounding the total size of the team. 
    The goal of this paper is to consider the team formation problem in a more realistic setting
 and  present a novel formulation based on a generalization of the densest subgraph problem. 
 Our formulation allows modeling of many realistic requirements such as
  (i) inclusion of  a designated team leader
 and/or  a group of given experts, 
 (ii) restriction on the size or more generally cost of the team
 (iii) enforcing \textit{locality} of the team, e.g., in a geographical sense or social sense, etc.
 In fact most of the future directions pointed out by \cite{GajSar12} are covered in our formulation.  
%

\section{Related Work}
 
 The first work \cite{LapLiuTer09}   in the team formation problem in the presence of a social network presents greedy 
 algorithms for minimizing the diameter and the cost of the minimum spanning tree (MST) induced by the team.
While the greedy algorithm for 
 minimizing the diameter 
 has an approximation guarantee of 
 two, no guarantee is proven for the MST algorithm. However, \cite{LapLiuTer09} impose the strong assumption that a skill requirement of a task can be fulfilled by 
 a single person; thus  a more natural requirement such as ``at least $k$ experts of skill $s$ are needed for the task'' cannot be handled by 
 their method. This shortcoming has been addressed in \cite{GajSar12}, which 
 presents a 2-approximation algorithm for a slightly more general problem that can accommodate the above requirement. 
 However, both algorithms cannot handle an upper bound 
 constraint on the team size.
 On the other hand, the solutions obtained by all these algorithms (including the MST algorithm) can be shown to be
   connected subgraphs if 
 the underlying social graph is connected. 
 
 Two new formulations are proposed in \cite{KarAn11} based on the shortest path distances between the nodes of the graph. 
 The first formulation assumes that experts from each skill have to communicate with every expert from the other skill and thus minimizes the 
 sum of the pairwise shortest path distances between experts belonging to different skills. 
 They prove that this problem is NP-hard and provide a greedy algorithm with an approximation guarantee of two. 
 The second formulation, solvable optimally in polynomial time, assumes that there is a designated team leader who 
 has to communicate with every expert in the team and minimizes the sum of the distances only to the leader. 
 The main shortcoming of this work is its restrictive assumption that {\it exactly} 
 one expert is sufficient for each skill,
 which implies that
 the size of the found teams is 
 always upper bounded by the number of skills in the given task, noting that an expert is allowed to have multiple skills. 
 They exploit this assumption and (are the first to) produce top-$k$ teams that can perform the given task.
  However, although based on the shortest path distances, neither of the two formulations does guarantee that the solution 
 obtained is connected.
 
  In contrast to the distance or diameter based cost functions, \cite{GajSar12} explore the usefulness of the density based objective
  in  finding strongly connected teams. 
      Using various performance measures,
  the superiority of the density based objective function over the diameter objective is demonstrated.
  The setting considered in \cite{GajSar12} is the most general one 
  until now
  but the resulting problem is shown to be NP hard. 
   The greedy algorithms that they propose 
  have approximation guarantees (of factor 3) for two special cases. 
  The teams found by their algorithms are often quite large and 
  it is not straightforward to modify their algorithms to integrate an additional 
  upper bound constraint on the team size.
  Another disadvantage is that subgraphs that maximize the density under the given constraints 
  need not necessarily be connected.

   Recently \cite{AnaBecCasGioLeo12} considered an {\it online} team formation problem where
   tasks arrive in a sequential manner and teams have to be formed minimizing the (maximum) load on any expert across the tasks
   while bounding the coordination cost (a free parameter)  within a team for any given task. Approximation algorithms are
   provided for two variants of coordinate costs: diameter cost and Steiner cost (cost of the minimum Steiner tree where the team 
   members are the terminal nodes).
   While this work focusses more on the load balancing aspect, it also makes the strong assumption that a skill is covered by the team 
   if there exists at least one expert having that skill.
   
  All of the above methods allow only binary skill level, i.e., an expert has a skill level of either one or zero. 
     
 We point out that many methods have been developed in the 
 operations research community for the team formation problem,
  \cite{BayDerDas07, CheLin04, ZzkKus04, WiOhMunJun09},
    but none of them explicitly considers the underlying social or professional connections 
 among the experts. There is also literature discussing the social aspects of the 
 team formation \cite{Contractor2013} and their influence on the evolution of communities, e.g., \cite{BacHutKleXia06}. 
 

\section{Realistic Team Formation in Social Networks}
 Now we formally define the {\it Team Formation} problem that we address in this paper.
 Let $V$ be the set of $n$ experts and $G(V,W)$ be the weighted, undirected graph reflecting the relationship or 
 previous collaboration of the experts $V$.
Then non-negative, symmetric weight $w_{ij} \in W$ connecting two experts $i$ and $j$ reflects the level of compatibility 
between them.  
The set of skills is given by $\A = \{a_{1}, \ldots, a_{p}\}$.
Each expert is assumed to possess one or more skills. The non-negative matrix $M \in \R^{n \times p}$ 
specifies the skill levels of all experts in each skill. Note that we define the skill level on a continuous scale.
 If an expert $i$ does not have skill $j$, then $M_{ij} = 0$. Moreover, we use the notation $M_j \in \R^ {n \times 1}$ for the $j-$th column of $M$, i.e.\, the vector of skill levels corresponding to skill $j$.
 A task $\T$ is given by the set of triples $\{(a_{j}, \kappa_{j}, \iota_{j})\}_{j=1}^{p}$, where $a_{j} \in \A$, specifying that 
 at least $\kappa_{j}$ and at most $\iota_{j}$ of skill $a_{j}$ is required to finish the given task. \\

\textbf{Generalized team formation problem.} 
Given a task  $\T$, the generalized team formation problem is defined as finding a team $C \subseteq V$ of experts 
maximizing the {\it collaborative compatibility} and 
satisfying the following constraints:
\begin{itemize} 
\item 
\textbf{Inclusion of a specified group:} a predetermined group of experts $S \subset V$ should be in $C$.
\item 
\textbf{Skill requirement:} at least $\kappa_{j}$ and at most $\iota_{j}$ of skill $a_{j}$ is required to finish the task $\T$. 
\item 
\textbf{Bound on the team size:} the size of the team should be smaller than or equal to $b$, i.e., $\abs{C} \le b$.
\item 
\textbf{Budget constraint:} total budget for finishing the task is bounded by $B$, i.e., 
					$\sum_{i \in C} c_{i} \le B$, where $c_{i} \in \R_+$ is the cost incurred on expert $i$.			
			
\item 
\textbf{Distance based constraint:} the distance  (measured according to some non-negative, symmetric function, $\dist$) between
	any pair of experts in $C$ should not be larger than $d_{0}$, 
	i.e., $\dist(u, v) \le d_{0}, \forall u, v \in C$.
\end{itemize}



\textbf{Discussion of our generalized constraints.}
In contrast to existing methods, we also allow an upper bound on each skill 
 and on the total team size. If the skill matrix is only allowed to be binary as in previous work, this translates into upper and lower bounds
on the number of experts required for each skill. 
Using vertex weights, we can in fact encode more generic constraints, e.g., having a limit on the total budget 
 of the team.
It is not straightforward to extend existing methods to include any upper bound constraints. 
Up to our knowledge we are the first to integrate upper bound constraints, in particular on the size of the team, 
into the team formation problem. 
We think that the latter constraint is essential for realistic team formation.

Our general setting also allows a group of experts around whom the team has to be formed. This constraint often applies
as the team leader is usually fixed before forming the team.
Another important generalization is the inclusion of {\it distance} constraints for any general distance function\footnote{The distance function need not satisfy the triangle inequality.}.  
Such a constraint can be used to enforce locality of the team e.g. in a geographical sense (the distance could be 
travel time) or social sense (distance in the network). Another potential application 
are mutual incompatibilities
of team members e.g.\;on a personal level,
which can be addressed
 by assigning a high distance to experts who are mutually incompatible
and thus should not be put together in the same team.



   We emphasize that all constraints considered in the  literature are special 
instances of the above constraint set.\\

\textbf{Measure of collaborative compatiblity.} 
 In this paper we use as a measure of collaborative compatibility a generalized form  of the density of subgraphs, defined as 
 \begin{align}\label{eq:gen_density}
	\density(C) := \frac{ \assoc(C) } { \vol_{g}(C)} =  \frac{ \sum_{i,j \in C} w_{ij} } {\sum _{i \in C} g_{i}},   
 \end{align}
 where $w_{ij}$ is the non-negative weight of the edge between $i$ and $j$  
  and $\vol_{g}(C)$ is defined as $\sum_{i \in C} g_{i}$, with $g_{i}$ being the positive weight of the vertex $i$.
 We recover the original density formulation, via $g_i=1, \forall i \in V$.
 We use the relation, $\assoc(C) = \vol_{d}(C) - \cut(C, V \backslash C)$, 
 where $d_{i}=\sum_{j=1}^n w_{ij}$ is the degree of vertex $i$ and 
 $\cut(A, B) := \sum_{i \in A, j \in B} w_{ij}$.
 \\

\textbf{Discussion of density based objective.}
 As pointed out in \cite{GajSar12}, the density based objective possesses useful properties like 
 strict monotonicity and robustness.
 In case of the density based objective, if an edge gets added (because of a new collaboration) 
 or deleted (because of newly found incompatibility) the density of the subgraphs involving this edge 
 necessarily increases resp.\;decreases, which is
 not true for the diameter based objective. 
 In contrast to density based objective, the impact of small changes in graph structure is more severe in 
 the case of diameter objective \cite{GajSar12}.

 The generalized density that we use here leads to further modeling freedom as it enables to give weights 
 to the experts according to their expertise. 
 By giving smaller weight to those with high expertise, one can obtain 
 solutions that not only satisfy the given skill requirements but also give preference to the more competent team members (i.e.\;the ones having smaller weights).\\
 
%
%
\textbf{Problem Formulation.}
Using the notation introduced above, an instance of the team formation problem 
 based on the generalized density 
 can be formulated as
 \begin{align}\label{eq:teamProb}
	\underset{C \subseteq V} \max &\; \frac{ \assoc(C)} {\vol_{g}(C)}  \\
	\subj: 	&\; S  \subseteq C \nonumber \\ 
		&\; \kappa_{j} \leq  \vol_{M_{j}}(C) \leq \iota_{j},\quad \forall j \in \{1, \ldots, p\} \nonumber \\	         
		 &\; \abs{C} \le b \nonumber \\
		 &\; \vol_{c} (C) \le B \nonumber \\		 
		 &\; \dist(u, v) \le d_{0}, \quad \forall u, v \in C, \nonumber
 \end{align}
 Note that the upper bound constraints on the team size and the budget can be 
 rewritten as skill constraints and can be incorporated into the skill matrix $M$ accordingly. 
 Thus, without loss of generality, we omit the budget and size constraints from now on, for the sake of brevity.
Moreover, since $S$ is required to be part of the solution, we can assume that 
$dist(u, v) \le d_{0}, \forall u, v \in S$, otherwise the above problem is infeasible.
 The distance constraint also implies that any $u \in V$ for which $\dist(u, s) > d_{0}$, for some $s \in S$, cannot be a part of the solution. 
 Thus, we again assume wlog that there is no such $u \in V$; otherwise such vertices can be eliminated without 
 changing the solution of problem \eqref{eq:teamProb}. 
\\

  Our formulation  \eqref{eq:teamProb} is a generalized version of the classical densest subgraph problem (DSP),
  which has many applications in graph analysis, e.g., see \cite{Sah2010}. 
  The simplest version of DSP is the problem of finding a densest subgraph (without any constraints on the solution), 
  which can be solved optimally in polynomial time \cite{Gol84}. 
  The densest-$k$-subgraph problem, which requires the solution to contain exactly $k$ vertices,
  is a notoriously hard problem in this class
  and has been shown not to admit a polynomial time approximation scheme \cite{Kho06}.
  Recently, it has been shown that the densest subgraph problem with an upper bound on the size is as hard as the 
  densest-$k$-subgraph problem \cite{KhuSah09}. 
   However, the densest subgraph problem with a lower bound constraint has a 2-approximation algorithm \cite{KhuSah09}.
  It is based on solving a sequence of unconstrained densest subgraph problems. 
   They also show that there exists a linear programming relaxation for this problem achieving the same approximation guarantee.
  
  Recently \cite{GajSar12} considered the following generalized version of the densest subgraph problem 
  with lower bound constraints in the context of team formation problem:
    \begin{align}\label{eq:teamProbLb}
	\underset{C \subseteq V} \max &\; \frac{ \assoc(C)} {\vol_{g}(C)}  \\
	\subj:
 		&\; \vol_{M_{j}}(C)  \geq \kappa_{j},\quad \forall j \in \{1, \ldots, p\} \nonumber
	 \nonumber          
 \end{align}
 where $M$ is the \textit{binary} skill matrix. 
 They extend the greedy method of \cite{KhuSah09} and show that it achieves a 3-approximation guarantee 
 for some special cases of this problem. \cite{Venk12} recently improved  the approximation guarantee of the greedy algorithm of \cite{GajSar12} for problem  (\ref{eq:teamProbLb}) to a factor 2.
 The time complexity of this greedy algorithm  
is $O(kn^{3})$, where $n$ is the number of experts and $k := \sum_{j=1}^{m}{k_{j}}$ is the minimum number of experts 
 required. 
\\

 \textbf{Direct integration of subset constraint.}
 The subset constraint can be integrated into the objective
 by directly working on the subgraph $\Gp$ induced by the vertex set $\Vp = V \backslash S$.
  Note that any $C \subset V$ that contains $S$ can be written as $C = A \cup S$, for $A \subset \Vp$.
 We now reformulate the team formation problem on the subgraph $\Gp$. We introduce the notation $m=\abs{\Vp}$, and we assume wlog that the first $m$ entries of $V$ are the ones in $\Vp$.

  The terms in problem (\ref{eq:teamProb}) can be rewritten as 
  \begin{align*}
    	\assoc(C) &= \assoc(A) + \assoc(S) + 2\ \cut(A,S), \\
    			& = \vol_d(A) -\cut(A,V \backslash A) + \assoc(S) + 2\ \cut(A,S) \\
    			& = \vol_d(A) -\cut(A,\Vc^{\prime} \backslash A) + \assoc(S) + \cut(A,S)\\
 	 \vol_{g}(C) &= \vol_{g}(A) + \vol_{g}(S) 
  \end{align*}
Moreover, note that we can write: $\cut(A,S)  = \vol_{d^{S}} (A)$, where $d^{S}_{i} = \sum_{j \in S} w_{ij}$ 
 denotes the degree of vertex $i$ restricted to the subset $S$ in the original graph.
 Using the abbreviations, 
$\mu_{S} = \assoc(S)$, $\nu_{S} = \vol_{g}(S)$, $\assoc_{S}(A) =  \vol_{d}(A) - \cut(A,\Vp \backslash A) + \mu_{S} +  \vol_{d^{S}} (A)$,
we rewrite the team formation problem \eqref{eq:teamProb} as 
  \begin{align}\label{eq:teamProbSim}
	\underset{A \subseteq \Vc^{\prime}, \ A \neq \emptyset} \max &\; 
		 \frac{ \assoc_{S}(A) } 
		 {\vol_{g}(A) + \nu_{S}}  \tag{GDSP} \\
	\subj: &\; k_{j}  \leq  \vol_{M_{j}}(A)  \leq l_{j},\quad \forall j \in \{1, \ldots, p\} \nonumber \\	         
		&\; \dist(u, v) \le d_{0},\quad \forall u, v \in A, \nonumber
 \end{align}
 where for all $j=1,\dots,p$, the bounds were updated as $k_j = \kappa_j-\vol_{M_j}(S),\ l_j = \iota_j-\vol_{M_j}(S)$. 
 Note that here we already used the assumption: $\dist(u, s) \le d_{0}, \forall u \in V, \forall s \in S$. The constraint, $A \neq \emptyset$, 
 has been introduced for technical reasons required for the formulation of the continuous problem in Section \ref{sec:equiv_cont_prob}. 
 The equivalence of problem \eqref{eq:teamProbSim} to  \eqref{eq:teamProb} follows by considering either $S$ (if feasible) 
or  the set $A^{*} \cup S$, where $A^{*}$ is an optimal solution of 
 \eqref{eq:teamProbSim}, depending on whichever has higher density.

 To the best of our knowledge there is no greedy algorithm with an approximation guarantee to solve 
 problem (\ref{eq:teamProbSim}).
 Instead of designing a greedy approximation algorithm for this discrete optimization problem, 
 we derive an \textit{equivalent}
 continuous optimization problem in Section \ref{sec:FORTE}.
 That is, we reformulate the discrete problem in continuous space 
 while preserving the optimality of the solutions of the discrete problem. 
 The rationale behind this approach is that the continuous formulation is more flexible
 and allows us to choose from a larger set of methods for its solution than for the discrete one.
Although the resulting continuous problem is as hard as the original discrete problem, 
recent progress in continuous optimization 
\cite{HeiSet11} 
allow us to find a locally optimal solution very efficiently.

  
\section{Derivation of FORTE}  \label{sec:FORTE}
 In this section we present our method, \textit{Formation Of Realistic Teams} (\FORTE, for short) to solve 
 the team formation problem, which is rewritten as \eqref{eq:teamProbSim},
 using the continuous relaxation. 
  We derive 
  \FORTE\, in three steps:
  \begin{enumerate}
\renewcommand{\theenumi}{\roman{enumi}}
\item Derive an equivalent unconstrained discrete problem (\ref{eq:teamProbUncstr}) of the team formation problem (\ref{eq:teamProbSim}) via an \textit{exact penalty} approach.
\item Derive an equivalent continuous relaxation (\ref{eq:teamProbCont}) of the unconstrained problem \eqref{eq:teamProbUncstr}
	by using the concept of \textit{Lovasz extensions}.
\item Compute the solution of the continuous problem (\ref{eq:teamProbCont}) using the recent method RatioDCA from \textit{fractional programming}.
\end{enumerate}

 \subsection{Equivalent Unconstrained Problem}
 A general technique in constrained optimization is to transform the constrained problem 
 into an equivalent unconstrained problem by adding to the objective a penalty term, which is 
 controlled by a parameter $\gamma \geq 0$. 
 The penalty term is zero if the constraints are satisfied at the given input and strictly positive otherwise. 
 The choice of the regularization parameter $\gamma$ influences the tradeoff between satisfying the constraints and having a low objective value.
Large values of $\gamma$ 
 tend to enforce the satisfaction of constraints. 
 In the following we show that for the team formation problem (\ref{eq:teamProbSim}) there exists a  value of 
 $\gamma$ that guarantees the satisfaction of all constraints. \\
 
Let us define the penalty term for constraints of the team formation problem (\ref{eq:teamProbSim}) as
	\begin{align*}
	 {\pen}(A) &:=  \left\{\begin{array}{ll} 
					 \sum_{j=1}^{p}\max\{0, \vol_{M_{j}}(A) - l_{j} \}\\					 	 + \sum_{j=1}^{p}\max\{0, k_j- \vol_{M_{j}}(A) \} \\
	 + \sum_{u, v \in A}\max\{0,\ \dist(u, v) - d_{0} \}
	 & A \neq \emptyset\\
	 0 & A = \emptyset. \end{array} \right.
	\end{align*}
	 Note  that the above penalty function is zero only when $A$ satisfies the constraints; otherwise it is
	 strictly positive and increases with increasing infeasibility.  	 
	 The special treatment of the empty set 
	 is again a technicality required later for the Lovasz extensions, see Section \ref{sec:equiv_cont_prob}.
	 For the same reason, we also replace the constant terms $\mu_{S}$ 
	 and $\nu_{S}$ in 	(\ref{eq:teamProbSim}) by $\mu_{S}\unit(A)$ and $\nu_{S}\unit(A)$ respectively, where 
	 $\unit(A) := 1, A \neq \emptyset$ and $\unit(\emptyset) = 0$.
	 \\

 The following theorem shows that there exists an unconstrained problem equivalent to the constrained optimization problem (\ref{eq:teamProbSim}).
 \begin{theorem} 	\label{thm:teamProbUncstr}
	 The constrained problem (\ref{eq:teamProbSim}) is equivalent to the unconstrained problem
	 \begin{align}\label{eq:teamProbUncstr}
	 	\underset{\emptyset \neq A \subseteq V} \min &\; 
				\frac
				 {\vol_{g}(A) + \nu_{S} \unit(A) + \gamma \pen(A)}
				 { \assoc_{S}(A) } 
 \end{align}
for $\gamma >  \frac{\vol_{d}(V)}{\theta}\; \frac {\vol_{g}(A_{0}) + \nu_{S} } { \assoc_{S}(A_{0})}$, 
where $A_{0}$ is any feasible set of problem (\ref{eq:teamProbSim}) such that $\assoc_{S}(A_{0}) > 0$ 
and $\theta$ is the minimum value of infeasibility, i.e., ${\pen} (A) \ge \theta$, if $A$ is infeasible.
 
 \begin{proof}
 We define 
 $\spvol(A) :=  \frac  {\vol_{g}(A) + \nu_{S} \unit(A)} { \assoc_{S}(A) }$.
  Note that maximizing (\ref{eq:teamProbSim}) is the same as minimizing
 $\spvol(A)$ subject to the constraints of (\ref{eq:teamProbSim}). 
  	For any feasible subset $A$, the objective of (\ref{eq:teamProbUncstr}) is equal to $\spvol(A)$,
	since the penalty term is zero.
	Thus, if we show that \textit{all} minimizers of
	(\ref{eq:teamProbUncstr}) satisfy the constraints then the equivalence follows. 
	Suppose, for the sake of contradiction, that $A^{*} (\neq \emptyset$, if $S = \emptyset$) is a minimizer of (\ref{eq:teamProbUncstr}) and that 
 	$A^*$ is infeasible for problem (\ref{eq:teamProbSim}). Since $\nu_{S} \ge 0$ and $g_{i} > 0, \forall i$,
	we have under the given condition on $\gamma$, 
		\begin{align*}	
		&\frac  {\vol_{g}(A^{*}) + \nu_{S} + \gamma\  \pen(A^{*}) } { \assoc_{S}(A^{*})}				
				> \frac{ 
				\gamma\  {\pen}(A^{*})} 
				{ \assoc_{S}(A^{*})}\\
				&\ge \frac{
				\gamma\  \theta} 
				{\max_{A \subseteq V} \assoc_{S}(A)}
				 \ge \frac{ 
				\gamma\  \theta} 
				{\vol_{d}(V)} > \frac  {\vol_{g}(A_{0}) + \nu_{S}} { \assoc_{S}(A_{0}) },
		\end{align*}	
		which leads to a contradiction because the last term is the objective value of (\ref{eq:teamProbUncstr}) at $A_{0}$.
\end{proof}			
\end{theorem}


\subsection{Equivalent Continuous Problem}\label{sec:equiv_cont_prob}
	We will now derive a tight continuous relaxation of problem (\ref{eq:teamProbUncstr}). 
	This will lead us to a minimization problem over $\R^m$, which then can be handled more easily than the original discrete problem.	
	The connection between the discrete and the continuous space is achieved via thresholding. 
	Given a vector $f\in \R^m$, one can define the sets	
\begin{equation}\label{eq:sets}
	A_i:= \left\{ j \in V | f_j \geq f_i\right\},
\end{equation}
	by thresholding $f$ at the value $f_{i}$. 
	In order to go from functions on sets to functions on continuous space, we make use of the concept of Lovasz extensions.
\begin{definition} {(Lovasz extension)}\label{def:lovasz}
Let $R:2^{V} \rightarrow \R$ be a set function with ${R}(\emptyset)=0$, and let $f \in \R^m$ be ordered in ascending order 
$f_1\leq f_2\leq\cdots\leq f_m$. 
The Lovasz extension $R^{L}:\R^m \rightarrow \R$ of ${R}$ is defined by
\begin{align*}
	R^{L}(f) 
	 = \sum_{i=1}^{m-1} {R}(A_{i+1}) \left(f_{i+1}-f_i\right) + {R}(V) f_1.
\end{align*}
\end{definition}
Note that $R^{L}(\ones_A)={R}(A)$ for all $A \subset V$, i.e.\;$R^{L}$ is indeed an extension of ${R}$ from $2^{V}$ to $\R^V$ ($|V|=m$).
In the following, given a set function $R$, we will denote its Lovasz extension by $R^{L}$. The explicit forms of the Lovasz extensions used in the derivation
will be dealt with in Section \ref{sec:algo}.




In the following theorem we show the equivalence for \ref{eq:teamProbSim}. 
A more general result showing equivalence for fractional set programs can be found in \cite{BueRanHeiSet13}.
\begin{theorem}\label{thm:teamProbCont}
	The unconstrained discrete problem (\ref{eq:teamProbUncstr}) is equivalent to the continuous problem
 	 \begin{align}\label{eq:teamProbCont}
	 	\underset{f \in \R_{+}^{\Vp}} \min &\;
				\frac
				 {  \vol_{g}^{L}(f) + \nu_{S} \unit^L(A) 
				 + \gamma \pen^{L}(f)}
				 { \assoc_{S}^{L}(f) } 		
	\end{align}
 	for any $\gamma \ge 0$. Moreover, optimal thresholding of a minimizer $f^{*} \in \R_{+}^{m}$, 
       \[ A^* := \min_{A_i = \left\{ j \in \Vp | f^*_j \geq f^*_i\right\}, \, i=1,\ldots,m} \frac
				 {\vol_{g}(A_i) + \nu_{S} + \gamma \pen(A_i)}
  			 { \assoc_{S}(A_i) },\]
	yields a set $A^{*}$ that is optimal for problem (\ref{eq:teamProbUncstr}).
\begin{proof}
	Let $R(A) = \vol_{g}(A) + \nu_{S} \unit(A) + \gamma \pen(A)$.
	Then we have 
	\[
		\min_{A \subset \Vp} \frac{R(A)}{\assoc_{S}(A)} 
		= \min_{A \subset \Vp} \frac{R^L(\ones_{A})}{ \assoc^{L}_{S}(\ones_{A})} 
		 \ge \min_{f \in \R_{+}^{\Vp}} \frac{R^{L}(f)}{\assoc^{L}_{S}(f)} ,
	\]
	where in the first step we used the fact that $R^{L}(f)$ and $\assoc^{L}(f)$ are extensions of $R(A)$ and $\assoc(A)$, respectively.
	Below we first show that the above inequality also holds in the other direction, which then establishes 
	that the optimum values of both problems are the same. The proof of the reverse direction will also imply that a set minimizer of the problem 
	(\ref{eq:teamProbUncstr}) can be obtained from any minimizer $f^{*}$ of (\ref{eq:teamProbCont}) via 
	optimal thresholding.\\
	
We first show that
	the optimal thresholding of any $f \in \R^{m}_{+}$ yields a set $A$ 
	such that $\ones_{A}$ has an objective value at least as good as the one of $f$.
	This holds because
	\begin{align*}
			R^{L}(f)
			&= \sum_{i=1}^{m-1} R(A_{i+1}) (f_{i+1} - f_{i}) +  f_{1} R(\Vp)\\
			&= \sum_{i=1}^{m-1} \frac{R(A_{i+1}) } {\assoc_{S}(A_{i+1})} \assoc_{S}(A_{i+1}) (f_{i+1} - f_{i})\\
				& \hspace{10mm} + \frac{R(\Vp)}{\assoc_{S}(\Vp)} \assoc_{S}(\Vp) f_{1}\\
			& \ge \min_{j=1, \ldots m} \frac{ R(A_{j}) } {\assoc_{S}(A_{j})} \\ 
				& \hspace{10mm} \Big(\sum_{i=1}^{m-1} \assoc_{S}(A_{i+1}) (f_{i+1} - f_{i})
				  + \assoc_{S}(\Vp) f_{1}\Big)\\
			&= \min_{j=1, \ldots m} \frac{ R(A_{j}) } {\assoc_{S}(A_{j})} \assoc_{S}^{L}(f) 
	\end{align*}	
	The third step follows from the fact that $f$ is non-negative ($f_{1} \ge 0$) and ordered in ascending order, i.e., $f_{i+1} - f_{i} \ge 0, \forall i=1,\dots,m-1$.
	Since $\assoc_{S}^{L}(f)$ is non-negative, the final step implies that
	\begin{align}\label{eq:optThres}
		 \frac{R^{L}(f)} {\assoc^{L}_{S}(f)} \ge \min_{j=1, \ldots m} \frac{ R(A_{j}) } {\assoc_{S}(A_{j})}  .
	\end{align}
	Thus we have 
	\begin{align*}
		\underset{f \in \R_{+}^{\Vp}} \min 
				\frac
				 {R^{L}(f)}
				 { \assoc_{S}^{L}(f) } 		
				 &\ge
				 \min_{A \subset \Vp} \frac{R(A)}
				 { \assoc_{S}(A) } 		.
	 \end{align*}
		 
	 From inequality (\ref{eq:optThres}), it follows that optimal thresholding of $f^{*}$ yields a set that 
	 is a minimizer of problem (\ref{eq:teamProbUncstr}).
\end{proof}
\end{theorem}

\begin{corollary}
	The team formation problem (\ref{eq:teamProbSim}) is equivalent to the problem (\ref{eq:teamProbCont}) 
	if $\gamma$ is chosen according to the condition given in Theorem \ref{thm:teamProbUncstr}.
\begin{proof}
	This directly follows from Theorems \ref{thm:teamProbUncstr} and \ref{thm:teamProbCont}.
\end{proof}
\end{corollary}
 
 While the continuous problem is as hard as the original discrete problem, recent ideas from continuous optimization \cite{HeiSet11} allow us to derive in the next section an algorithm for obtaining locally optimal solutions very efficiently.
 
  \subsection{Algorithm for the Continuous Problem}  \label{sec:algo}
  We now describe an algorithm for (approximately) solving the continuous optimization problem (\ref{eq:teamProbCont}).
  The idea is to make use of the fact that the fractional optimization problem (\ref{eq:teamProbCont}) has a special structure: as we will show in this section, it can be written as a special ratio of difference of convex (d.c.) functions, i.e.\, it has the form
    \begin{align}\label{eq:fracSet}
 	\min_{f \in \R_{+}^{V}} \frac {R_{1}(f) - R_{2} (f)}{S_{1}(f)  -S_{2}(f)} := Q(f),
 \end{align}
 where the functions $R_{1}, R_{2}, S_{1}$ and $S_{2}$ are positively one-homogeneous convex functions\footnote{A function $f$ is said to be positively one-homogeneous if $f(\alpha x) = \alpha f(x), \alpha \ge 0$.} and numerator and denominator are nonnegative. 
 This reformulation then allows us to use a recent first order method called RatioDCA \cite{HeiSet11, BueRanHeiSet13}.
  
  In order to find the explicit form of the convex functions, we first need to rewrite the penalty term as $\pen(A)=\pen_1(A)-\pen_2(A)$, where 
  \begin{align*}
  	 \pen_1\hskip-0.04cm(A)  & =  \textstyle{\sum_{j=1}^p} \hskip-0.1cm \vol_{M_j} (A) + \textstyle{\sum_{j=1}^p} k_j \unit(A),\\
  	 \pen_2\hskip-0.04cm(A)  & =   \textstyle{\sum_{j=1}^p} \hskip-0.1cm \min\{l_j, \hskip-0.05cm\vol_{M_{j}}\hskip-0.05cm(A) \}  
  	\hskip-0.1cm + \hskip-0.1cm\sum_{j=1}^p \hskip-0.1cm\min\{k_j,\hskip-0.05cm \vol_{M_{j}}\hskip-0.05cm(A) \} 
  	  \hskip-0.1cm \\
  	  & -   \textstyle{\sum_{u,v \in A}} \hskip-0.08cm  \max\{0,\ \dist(u, v) - d_{0} \}.
  \end{align*}
 Using this decomposition of $\pen(A)$, we can now write down the functions $R_{1}, R_{2}, S_{1}$ and $S_{2}$ as   
  \begin{align*}
    R_{1}(f) & = \vol_{\rho}^{L}(f) + \sigma \max_i\{f_i\}\\	
	R_{2}(f) &=  \gamma \pen_2^{L}(f)  \\
	S_{1}(f) &= \vol_{d}^{L}(f) + \vol_{d^{S}}^{L} (f) + \mu_{S} \max_{i}\{ f_i \} \\
	S_{2}(f) &= \cut^{L}(f). 
 \end{align*}
 where $\rho:= g + \gamma \sum_{j=1}^p M_j$, $\sigma:= \nu_S + \gamma \sum_{j=1}^p k_j$, $\pen_2^L(f)$ denotes the Lovasz extension of $\pen_2(A)$, and
  \begin{align*}
 \vol^{L}_{h}(f) &= \inner{(h_i)_{i=1}^m, f}, \text{where } h\in \R^n,\\
 	\cut^{L}(f) & = \textstyle{\frac{1}{2} \sum_{i, j =1}^{m} w_{ij} \abs{f_{i} - f_{j}}}.  
 \end{align*}
 \begin{lemma}
 Using the functions $R_{1}, R_{2}, S_{1}$ and $S_{2}$ defined above, the problem (\ref{eq:teamProbCont}) can be rewritten in the form \eqref{eq:fracSet}. The functions  $R_{1}, R_{2}, S_{1}$ and $S_{2}$ are  convex and positively one-homogeneous, and $R_{1}-R_{2}$ and $S_{1}-S_{2}$ are nonnegative.
 \end{lemma}
 \begin{proof}
The denominator of  \eqref{eq:teamProbCont} is given as
  $
  \assoc^{L}_{S}(f)= \vol^{L}_{d}(f) - \cut^{L}(f)+ \vol^{L}_{d^S}(f) + \mu_{S} \unit^L(f), 
  $
  and the numerator is given as
  $
  \vol_{g}^{L}(f) + \nu_{S} \unit^L(A) 
				 + \gamma \pen^{L}(f) .
  $
  Using Prop.2.1 in \cite{Bach11} and the decomposition of $\pen(A)$ introduced earlier in this section, we can decompose $\pen^{L}(f) = \pen_1^{L}(f) - \pen_2^{L}(f)$. The Lovasz extension of $\pen_1(A)$ is given as
  $
  \pen_1^L(f) = \sum_{j=1}^p \vol_{M_j}^{L}(f) + \sum_{j=1}^p k_j \max_i\{f_i\},
  $
  and let $\pen_2^L(f)$ denote the Lovasz extension of $\pen_2(A)$ (an explicit form is not necessary, as shown later in this section). The equality between \eqref{eq:teamProbCont} and \eqref{eq:fracSet} then follows by simple rearranging of the terms.
  
 The nonnegativity of the functions $R_{1}- R_{2}$ and $S_{1} -S_{2}$ follows from the nonnegativity of denominator and numerator of \eqref{eq:teamProbCont} and the definition of the Lovasz extension.
  Moreover, the Lovasz extensions of any set function is positively one-homogeneous \cite{Bach11}. 
  
 Finally, the convexity of $R_1$ and $S_1$ follows as they are a non-negative combination of the convex functions $\max_i\{f_i\}$ and $\inner{(h_i)_{i=1}^m,f}$ for some $h\in \R^n$. The function $S_{2}(f)=\cut^L(f)$ is well-known to be convex \cite{Bach11}. To show the convexity of $R_{2}$, we will show that the function $\pen_2(A)$ is submodular\footnote{A set function $R:2^V \rightarrow \R$ is submodular if for all $A,B \subset V$,
	$R(A\cup B) + R(A \cap B)  \leq R(A) + R(B).$ }. 
The convexity then follows from the fact that a set function is submodular if and only if its Lovasz extension is convex \cite{Bach11}.
For the proof of the submodularity of the first two sums one uses the fact that the pointwise minimum of a constant and a increasing submodular function is again submodular. Writing $D_{uv}:= \max\{0,\ \dist(u, v) - d_{0} \}$, the last sum can be written as 
$
 - \sum_{u,v\in A} D_{uv} = - \sum_{u\in A,v\in \Vp} D_{uv}
 + \sum_{u\in A,v\in \Vp \backslash A} D_{uv}
$.
Using $(d_D)_i=\sum_j D_{ij}$, we can write its Lovasz extension as
$
   - \vol_{d_D}(f) + 
   \frac{1}{2} \sum_{i,j\in \Vp} D_{ij}\abs{f_i-f_j}, 
$
which is a sum of a linear term and a convex term.
 \end{proof}

The reformulation of the problem in the form \eqref{eq:fracSet} enables us to apply a modification of the recently proposed RatioDCA \cite{HeiSet11, BueRanHeiSet13}, a method for the 
  \emph{local} minimization of objectives of the form \eqref{eq:fracSet} on the whole $\R^m$. 
\floatname{algorithm}{}
\begin{algorithm}[htb]
   \renewcommand{\thealgorithm}{}
   \caption{\textbf{RatioDCA \cite{HeiSet11}} Minimization of a non-negative ratio of one-homogeneous d.c functions over $\R^m_+$}
   \label{alg:ratio_dc}
\begin{algorithmic}[1]
   \STATE {\bfseries Initialization:} $f^0 \in \R^m_+$, 
   $\lambda^0 = Q(f^0)$
   \REPEAT
   \STATE  $f^{l+1} =\hspace{-4mm}
   	 \argmin_{u \in \R^{m}_{+},\, \norm{u}_2 \leq 1}\hspace{-3.5mm} R_1(u) +\lambda^l S_2(u)- \inner{u, r_2(f^l)+ \lambda^l s_1(f^l)}$\\
   	 	\text{where\ } $r_2(f^l) \in \partial R_2(f^l)$, $s_1(f^l) \in \partial S_1(f ^l)$
   \STATE $\lambda^{l+1}= Q(f^{l+1})$
	\UNTIL $\frac{\abs{\lambda^{l+1}-\lambda^l}}{\lambda^l}< \epsilon$
\end{algorithmic}
\end{algorithm} 
Given an initialization $f_{0}$, the above algorithm solves a sequence of convex optimization problems (line 3).
Note that we do not need an explicit description of the terms $S_1(f)$ and $R_2(f)$, but only elements of their sudifferential $s_1(f) \in \partial S_1(f)$ resp.\, $r_2(f) \in \partial R_2(f)$. The explicit forms of the subgradients are given in the appendix. 
The convex problem (line 3) then has the form
\begin{align}\label{eq:inner}			
\min_{f \in \R^{m}_{+}}  \frac{\lambda^{l}}{2} \sum_{i, j =1}^{m} w_{ij} \abs{f_{i} - f_{j}} + 
				\inner{f, c}
				+ \sigma\ \max_{i} \{f_{i} \},
\end{align}				 
where $c = \rho -  r_{2}(f^{l}) - \lambda^{l} s_{1}(f^{l})$. 
Note that \eqref{eq:inner} is a \textit{non-smooth} problem. However,  there exists an equivalent smooth dual problem, which we give below. 
\begin{lemma}\label{lem:inner_problem}
The problem \eqref{eq:inner} is equivalent to
\begin{small}
\begin{align*}
\min_{\substack{\norm{\alpha}_\infty\leq 1\\ \alpha_{ij}=-\alpha_{ji}}} 
	 \min_{v \in S_m} 
		\frac{1}{2}\norm{P_{\R_+^m} 
		\left(\hskip-0.1cm -c -\frac{\lambda^l}{2} A \alpha - \sigma v \hskip-0.1cm \right)}_2^2 \hskip-0.1cm , 
\end{align*}
\end{small}
where $A:\R^E \mapsto \R^V$ with $(A\alpha)_i := \sum_{j} w_{ij} (\alpha_{ij}-\alpha_{ji})$, $P_{\R_+^m}$ denotes the projection on the positive orthant
and $S_m$ is the simplex $S_m=\{v \in \R^m \,|\, v_i\geq 0, \sum_{i=1}^m v_i=1\}$.
\end{lemma}
\begin{proof}
First we use the homogenity of the objective in  the inner problem to eliminate the norm constraint. This yields
the equivalent problem
\begin{align*}
\min_{u\in \R^n_+}\sigma \max_i u_i  +\frac{1}{2}\norm{u}_2^2 + \inner{u,c}  + \frac{\lambda^l}{2}\hspace*{-0.1cm}\sum_{i,j=1}^n w_{ij}|u_i-u_j|.
\end{align*}
We derive the dual problem as follows:
\begin{align*}
		 & \min_{u\in \R_+^n}\hspace*{-0.1cm}\frac{\lambda^l}{2}\hspace*{-0.1cm}\sum_{i,j=1}^{n} w_{ij} \abs{u_{i} - u_{j}} + \sigma \max u_i + \frac{1}{2} \norm{u}_2^2+ \inner{u,c}  \\
		 & = \min_{u\in \R_+^n} \Big\{\max_{\substack{\norm{\alpha}_\infty\leq 1\\ \alpha_{ij}=-\alpha_{ji}}}\frac{\lambda^l}{2}\sum_{i,j=1}^{n} w_{ij} \left(u_{i} - u_{j}\right) \alpha_{ij}    \\
		 & \hspace{1.5cm}
		   +   \max_{v \in S_n} \sigma \inner{u,v} + \frac{1}{2} \norm{u}_2^2 + \inner{u,c} \Big\}  \\
		 & = \max_{\substack{\norm{\alpha}_\infty\leq 1 \\ \alpha_{ij}=-\alpha_{ji} \\ v \in S_n}} 
		 	\min_{u\in \R_+^n} \frac{1}{2} \norm{u}_2^2 + \inner{u, c + \frac{ \lambda^l}{2} A\alpha + \sigma v},		  
\end{align*}
where $(A\alpha)_i := \sum_{j} w_{ij} (\alpha_{ij}-\alpha_{ji})$.
The optimization over $u$ has the solution
$
	u= P_{\R_+^n} ( -c  - \frac{\lambda^l}{2} A \alpha - \sigma v).
$
Plugging $u$ into the objective and using that 
$\langle P_{\R_+^n}(x),x \rangle = \|P_{\R_+^n}(x)\|_2^2$, we obtain the result.
\end{proof}

The smooth dual problem can be solved very efficiently using recent scalable first order methods like FISTA \cite{BT09}, 
which has a guaranteed convergence rate of $O(\frac{1}{k^{2}})$, where $k$ is the number of steps done in FISTA.
The main part in the calculation of FISTA consists of a matrix-vector multiplication. As the social network is typically sparse, this 
operation costs $O(m)$, where $m$ is the number of non-zeros of $W$. 

RatioDCA \cite{HeiSet11}, produces a strictly decreasing sequence $f^{l}$, i.e., $Q(f^{l+1}) < Q(f^{l})$, or terminates. 
This is a typical property of fast local methods in non-convex
optimization. 
Moreover, the convex problem need not be solved to full accuracy; we can terminate the convex problem early, if the current $f^{l}$ produces already sufficent descent in $Q$. As the number of required steps in the 
RatioDCA typically ranges between 5-20, the full method scales to large networks.
Note that convergence to the global optimum of (\ref{eq:fracSet}) cannot be guaranteed 
due to the non-convex nature of the problem.
However, we have the following quality guarantee for the team formation problem.

\begin{theorem}\label{th:quality-guarantee} 
Let $A_{0}$ be a feasible set for the problem \eqref{eq:teamProbSim} and 
$\gamma$ is chosen as in Theorem \ref{thm:teamProbUncstr}.
Let $f^*$ denote the result of RatioDCA after initializing with the vector $\ones_{A_{0}}$, and let $A_{f^*}$ denote the set found by optimal thresholding of $f^*$. 
Either RatioDCA terminates after one iteration, or produces $A_{f^*}$ which satisfies 
all the constraints of the team formation problem (\ref{eq:teamProbSim}) and 
\[ 
	\frac{ \assoc_{S}(A_{f^{*}}) } {\vol_{g}(A_{f^{*}}) + \nu_{S} } > \frac{ \assoc_{S}(A_{0}) } {\vol_{g}(A_{0}) + \nu_{S} } .
 \]
 \begin{proof}
 	RatioDCA generates a decreasing sequence $\{f^{l}\}$ such that $Q(f^{l+1}) < Q(f^{l})$ until it terminates \cite{HeiSet11}. 
	We have $Q(f^{1}) < Q(\ones_{A_{0}})$, if the algorithm does not stop in one step.
	As shown in Theorem (\ref{thm:teamProbCont}) optimal thresholding of $f^{1}$ yields a 
	set $A_{f}$ that achieves smaller objective on the corresponding set function. 
	Since the chosen value of $\gamma$ guarantees the satisfaction of the constraints,  $A_f$ has to be feasible.
 \end{proof}
\end{theorem}

\section{LP relaxation of GDSP}\label{sec:LPrel}	
	Recall that our team formation problem based on the density objective
	is rewritten as the following GDSP after integrating the subset constraint:
	    \begin{align}\label{eq:teamProbSim2}
	\underset{A \subseteq \Vc^{\prime}} \max &\; 
		 \frac{ \assoc_{S}(A) } 
		 {\vol_{g}(A) + \nu_{S}}  \\
	\subj: &\; k_{j}  \leq  \vol_{M_{j}}(A)  \leq l_{j},\quad \forall j \in \{1, \ldots, p\} \nonumber \\	         
		&\; \dist(u, v) \le d_{0},\quad \forall u, v \in A \nonumber	
 \end{align}
 	Note that here we do not require the additional constraint, $A \neq \emptyset$, that we added to \eqref{eq:teamProbSim}.
 	In this section we show that there exists a Linear programming (LP) relaxation for this problem. 
	The LP relaxation can be solved optimally in polynomial time and provides 
	an upper bound on the optimum value of GDSP.
	In practice such an upper bound is useful to check the quality of the solutions found
	by approximation algorithms. \\
	
	\begin{theorem}
	The following LP is a relaxation of the Generalized Densest Subgraph Problem (\ref{eq:teamProbSim2}).
		\begin{align}\label{eq:LP}
		\underset{t \in \R,\  f \in \R^{V^{\prime}},\ \alpha \in \R^{E^{\prime}}}  \max &\; 
		  \sum_{i,j=1}^{m} w_{ij} \alpha_{ij} + 2\ \inner{d^{S}, f} + t \mu_{S}\\		 \nonumber
		\subj: &\; tk_{j}  \leq  \inner{M_{j}, f}  \leq tl_{j},\quad \forall j \in \{1, \ldots, p\} \\	     \nonumber
		 &\; f_{u} + f_{v} \le t,\quad \forall u, v : \dist(u, v) > d_{0}\\ \nonumber
		 &\; t \ge 0, \quad \alpha_{ij} \le f_{i},\ \alpha_{ij} \le f_{j}, \ \forall (i,j) \in E^{\prime}\\ \nonumber
		 &\; 0 \le  f_{i} \le t, \ \forall i \in \Vc^{\prime},\
		  \alpha_{ij} \ge  0, \  \forall (i, j) \in E^{\prime}\\ \nonumber
		 &\; {\inner{g, f} + t\nu_{S}}  = 1. \nonumber
	 \end{align}
	where $V^{\prime} = V \backslash S$,
	 $E^{\prime}$ is the set of edges induced by $V^{\prime}$.
	\begin{proof}
	The following problem is equivalent to (\ref{eq:teamProbSim2}), because (i) for every feasible set $A$ of \eqref{eq:teamProbSim2}, 
	there exist corresponding feasible $y,\ X$ given by $y= \ones_{A}$, $ X_{ij} = \min\{ y_{i}, y_{j}\}$, with the same 
	objective value and (ii) an optimal solution of the following problem always satisfies $X^{*}_{ij} = \min \{y^{*}_{i}, y^{*}_{j}\}$.
		    \begin{align*}
		\underset{y \in \{0,\ 1\}^{V^{\prime}},\ X \in \{0,\ 1\}^{E^{\prime}}}  \max &\; 
		 \frac{ 2\sum_{i<j} w_{ij} X_{ij} + 2\ \inner{d^{S}, y} + \mu_{S}}
		 {\inner{g, y} + \nu_{S}}  \\
	\subj: &\; k_{j}  \leq  \inner{M_{j}, y}  \leq l_{j},\quad \forall j \in \{1, \ldots, p\} \\	         
		 &\; y_{u} + y_{v} \le 1,\quad \forall u, v : \dist(u, v) > d_{0}\\ \nonumber
		&\; X_{ij} \le y_{i},\quad X_{ij} \le y_{j}, \quad \forall (i,j) \in E^{\prime}
		  \end{align*}

	Relaxing the integrality constraints and using the substitution, 
	$ X_{ij} = \frac{\alpha_{ij} }{t}$ and $y_{i} = \frac{f_{i}}{t}$, 
	we obtain the relaxation:
	\begin{align*}
		\underset{t\in \R,\ f \in \R^{V^{\prime}},\  \alpha \in \R^{E^{\prime}}}  \max &\; 
		 \frac{ 		2 \sum_{i<j} w_{ij} \alpha_{ij} + 2\ \inner{d^{S}, f} + t \mu_{S}}
		 {\inner{g, f} + t\nu_{S}}  \\
		\subj: &\; tk_{j}  \leq  \inner{M_{j}, f}  \leq tl_{j},\quad \forall j \in \{1, \ldots, p\} \\	         
		 &\; f_{u} + f_{v} \le t,\quad \forall u, v : \dist(u, v) > d_{0}\\ \nonumber
		 &\; t \ge 0, \quad \alpha_{ij} \le f_{i},\ \alpha_{ij} \le f_{j}, \ \forall (i,j) \in E^{\prime}\\
		 &\; 0 \le  f_{i} \le t, \ \forall i \in \Vc^{\prime},\
		  \alpha_{ij} \ge  0, \  \forall (i, j) \in E^{\prime}
	 \end{align*}
	
	Since this problem is invariant under scaling, we can fix the  scale 
	by setting the denominator to 1, which yields the equivalent LP stated in the theorem.
	\end{proof}
	\end{theorem}	
 	 	
	Note that the solution $f^{*}$ of the LP (\ref{eq:LP}) is, in general, not integral, i.e.,
	$f^{*} \notin \{0, 1\}^{V^{\prime}}$.
	One can use  standard techniques of randomized rounding or optimal thresholding 
	to derive an integral solution from $f^{*}$. 
	However, the resulting integral solution may not necessarily give a subset that satisfies
	 the constraints of \eqref{eq:teamProbSim2}. In the special case when there are only lower bound constraints, 
	 i.e., problem (\ref{eq:teamProbLb}), one can obtain a feasible set $A$ for problem (\ref{eq:teamProbLb})
	  by thresholding $f^{*}$  (see \eqref{eq:sets}) according to the objective of \eqref{eq:teamProbSim2}. 
	 This is possible in this special case because there is always a threshold $f^{*}_{i}$ which yields 
	 a non-empty subset $A_{i}$ (in the worst case the full set $V^{\prime}$) satisfying all the lower bound constraints. 
	 In our experiments on problem \eqref{eq:teamProbLb}, we derived a feasible set from the solution of LP in this fashion
	 by choosing the threshold that yields a subset that satisfies the constraints and has the highest objective value.	 
	 
	 Note that the LP relaxation (\ref{eq:LP}) is vacuous with respect to 
	 upper bound constraints 
	 in the sense  that given $f \in \R^{m}$ that does not satisfy 
	 the upper bound constraints of the LP (\ref{eq:LP}) one can construct $\tilde{f}$, feasible for the LP 
	 by rescaling $f$ without changing the objective of the LP. This implies that one can always transform the solution of the unconstrained problem into a feasible solution
	 when there are $only$ upper bound constraints.
	 However, in the presence of lower bound or subset constraints, such 
 	 a rescaling does not yield a feasible solution and hence the LP relaxation is useful
	 on the instances of \eqref{eq:teamProbSim2} with at least one lower bound or a subset constraint (i.e., $\nu_{S} > 0$). 
	 
%
%

  \section{Experiments}
  
  We now empirically show that \FORTE\; consistently produces high quality 
  compact teams. 
  We also show that the quality guarantee given by Theorem \ref{th:quality-guarantee} is  
  useful in practice as our method 
  often improves a given sub-optimal solution.
  
  \subsection{Experimental Setup}
   Since we are not aware of any publicly available real world datasets for the team formation problem, we use, as in \cite{GajSar12}, a scientific collaboration network extracted from the DBLP database. 
  Similar to \cite{GajSar12}, we restrict ourselves to
   four fields 
  of computer science: Databases (\DB), Theory (\Th), Data Mining (\DM), Artificial Intelligence (\AI). 
  Conferences that we consider for each field are given as follows:
 {\small DB = \{SIGMOD, VLDB, ICDE, ICDT, PODS\}, 
	 T = \{SODA, FOCS, STOC, STACS, ICALP, ESA\}, 
	 DM = \{WWW, KDD, SDM, PKDD, ICDM, WSDM\}, 
	 AI = \{IJCAI, NIPS, ICML, COLT, UAI, CVPR\}}. 	 
 
 For our team formation problem, the skill set is given by $\A = $\{\DB, \Th, \DM, \AI\}. 
 Any author who has at least three publications in any of the above 23 conferences is 
 considered to be an expert.
 In our DBLP co-author graph, a vertex corresponds to an expert and 
 an edge between two experts indicates prior collaboration between them. 
 The weight of the edge is the number of shared publications. 
 Since the resulting co-author graph is disconnected, we take its largest connected component 
 (of size 9264)
 for our experiments. 
 
   Directly solving the non-convex problem (\ref{eq:teamProbCont}) for the value of $\gamma$ given in 
  Theorem \ref{thm:teamProbUncstr} 
   often yields poor results. 
 Hence in our implementation of \FORTE\ we adopt the following strategy. 
  We first solve the unconstrained version of problem (\ref{eq:teamProbCont}) (i.e., $\gamma=0$) 
 and then iteratively solve (\ref{eq:teamProbCont}) for increasing values of $\gamma$ until all constraints are satisfied.
  In each iteration, we increase $\gamma$ only for those constraints which were 
  infeasible in the previous iteration;  in this way, each penalty term is regulated by different value of $\gamma$.
  Moreover, the solution obtained in the previous iteration of $\gamma$ is used as the starting point for the current iteration. 

  \subsection{Quantitative Evaluation}
\begin{figure*}[t]
\centering
	\includegraphics[scale=0.27]{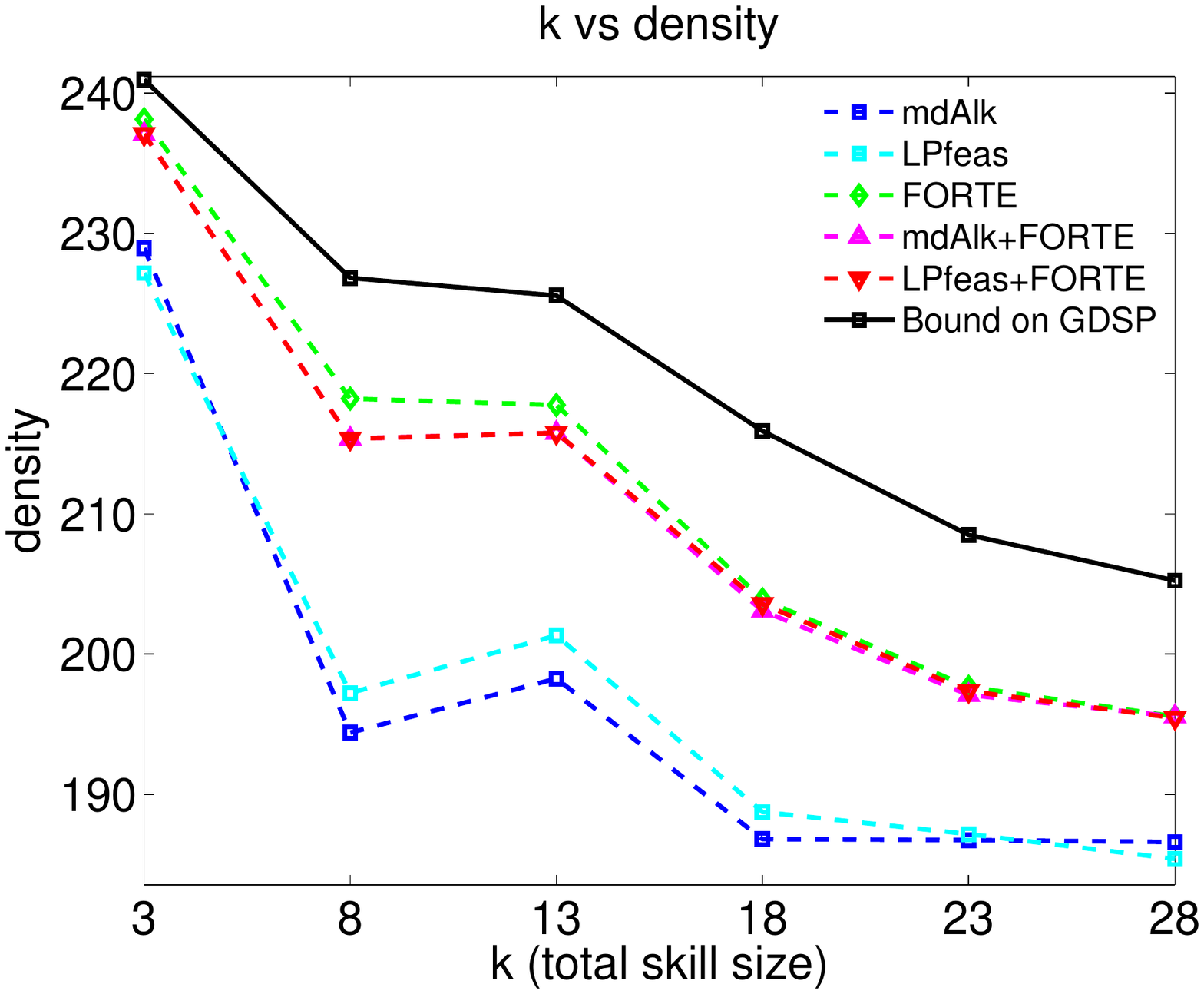}		
	\includegraphics[scale=0.27]{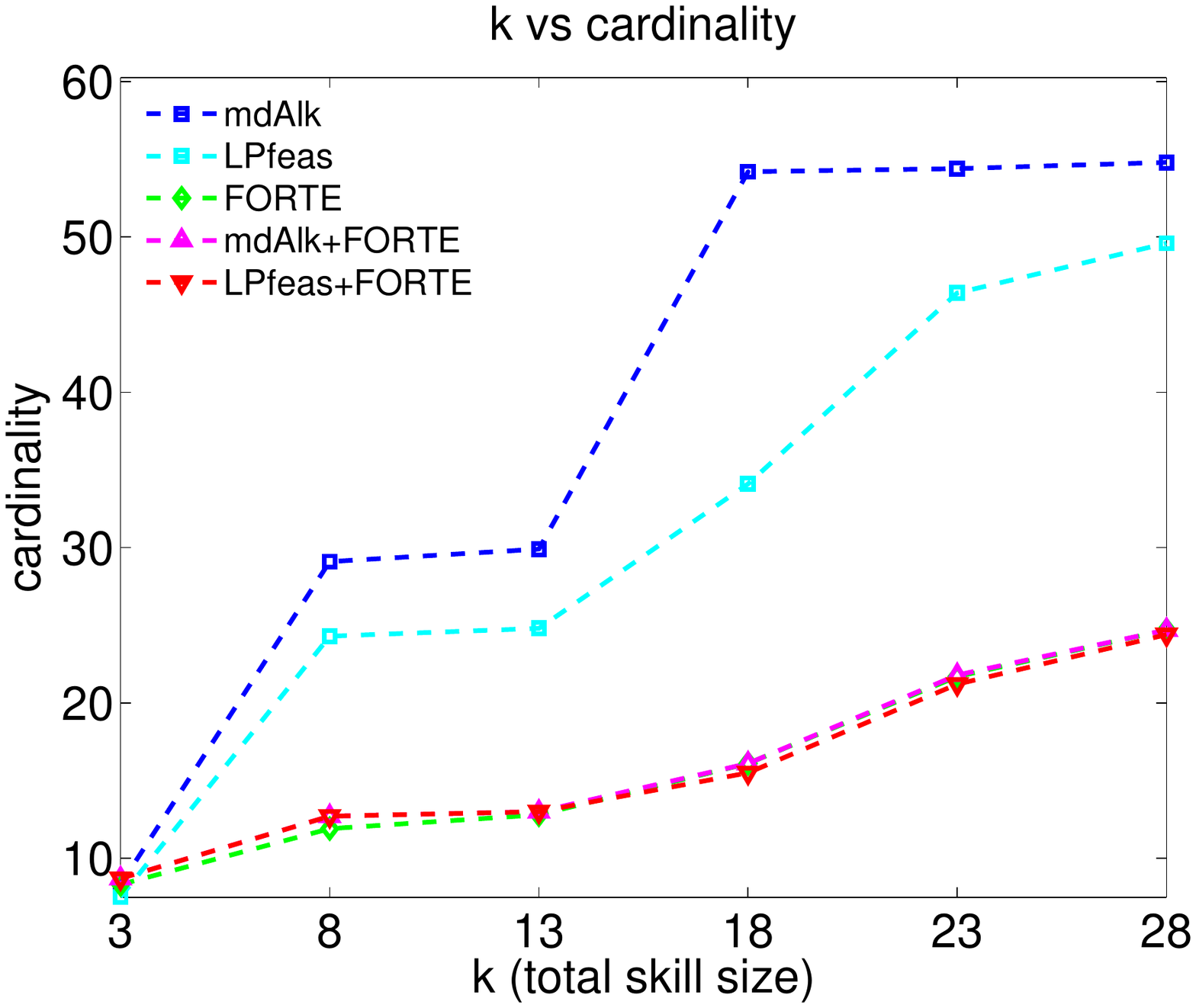}		
	\includegraphics[scale=0.27]{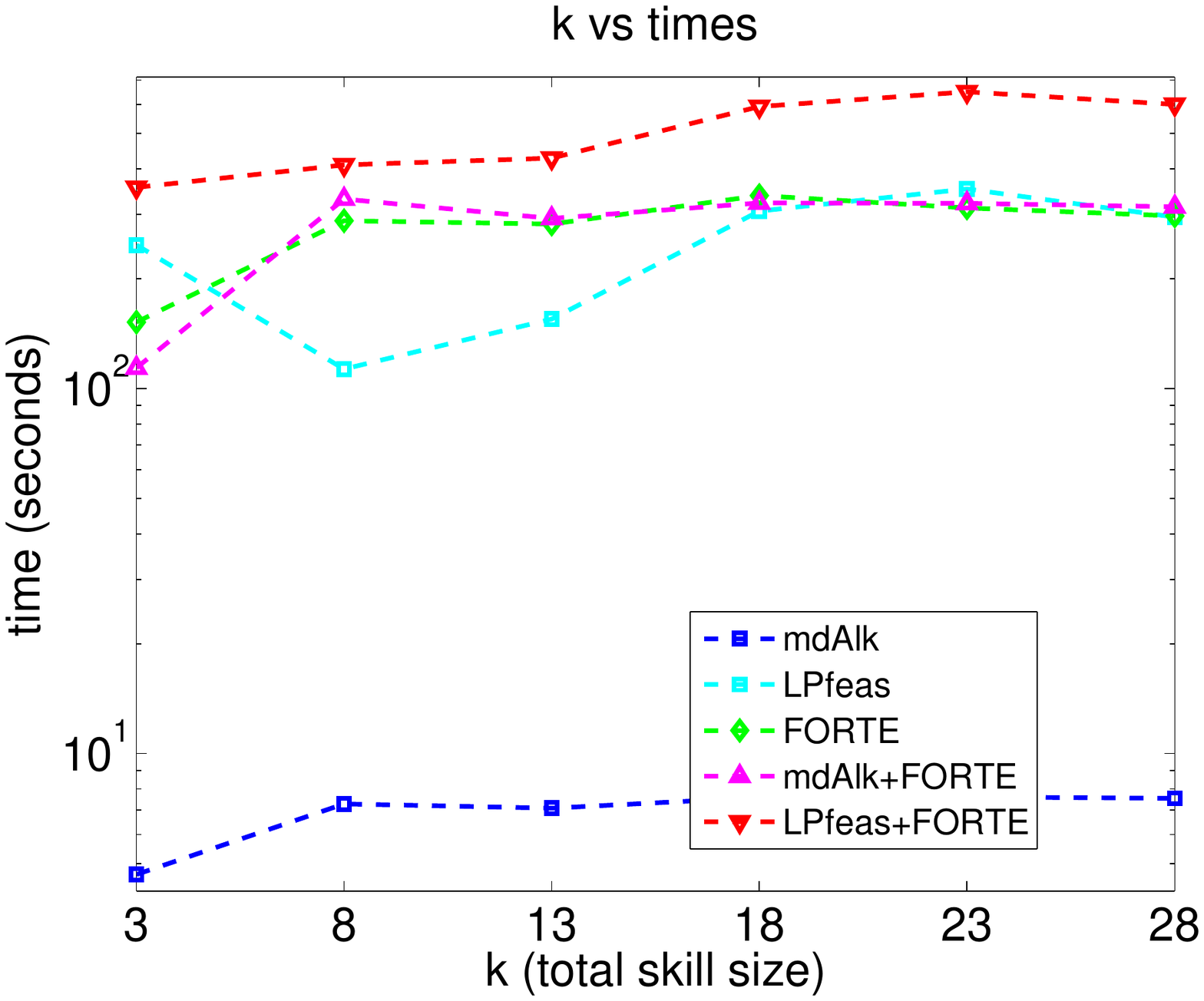}				
\label{fig:times}
\vskip -0.2cm
\caption{Densities, team sizes and runtimes of \mdalk, our method (\FORTE), a feasible point constructed from the LP  (\lpfeas), and \FORTE\; initialized with \lpfeas\; and \mdalk,  averaged over 10 trials. All versions of  (\FORTE) significantly outperform \mdalk, and LPfeas  both in terms of densities and 
sizes of the teams found. The densities of \FORTE\ are close to the upper bound on the optimum of the GDSP given by the LP.}
\end{figure*}



%
 In this section we perform a quantitative evaluation of our method in the special case of the team formation problem with lower bound constraints and $g_i=1 \, \forall i$ (problem \eqref{eq:teamProbLb}).
  We evaluate the performance of our method against the greedy method proposed in \cite{GajSar12}, refered to as \mdalk. 
 Similar to the experiments of \cite{GajSar12}, an expert is defined to have a skill level of 1 in skill $j$, if he/she has a publication in any of the conferences 
 corresponding to the skill $j$.
 As done in \cite{GajSar12}, we create random tasks for different values of skill size, $k = \{3, 8, 13, 18, 23, 28 \}$.
 For each value of $k$ we sample $k$ skills with replacement from the skill set $\A$ = \{\DB, \Th, \DM, \AI\}.
 For example if $k=3$, a sample might contain \{\DB, \DB, \Th \}, which means that the random task 
 requires at least two experts from the skill \DB\ and one expert from the skill \Th. 

In Figure 1, we show for each method the densities, sizes and runtimes for the different skill sizes $k$, averaged over 10 random runs. In the first plot, we also show the optimal values of the LP relaxation in \eqref{eq:LP}.  Note that this provides an upper bound on the optimal value 
of \eqref{eq:teamProbSim}.
 We can obtain  feasible solutions from the LP relaxation of \eqref{eq:teamProbSim} via thresholding (see Section \ref{sec:LPrel}),
 which are shown in the plot as \lpfeas.
 Furthermore, the plots contain the results obtained when the solutions of \lpfeas\ and \mdalk\ 
 are used as the initializations for
 \FORTE\  (in each of the $\gamma$ iteration). 
 
 The plots show that \FORTE\ always produces teams of higher densities
 and smaller sizes  compared to \mdalk\ and \LPFeas. 
Furthermore, 
  \lpfeas\; produces better results than 
 the greedy method in several cases in terms of densities and sizes of the obtained teams. 
 The results of \mdalk+\FORTE\; and \lpfeas+\FORTE\; 
 further show that our method is able improve the sub-optimal solutions of \mdalk\ and \lpfeas\ significantly and 
 achieves almost similar results as that of \FORTE\, which was started with the unconstrained solution of \eqref{eq:teamProbCont}. 
   Under the worst-case assumption that the upper bound on \eqref{eq:teamProbSim} computed using the LP is the optimal value, the solution of \FORTE\; is $94\%-99\%$ optimal (depending on $k$).
\subsection{Qualitative Evaluation}
\begin{table*}[htb!]\label{tb:qualTeams}
\begin{tabular}{|p{2.5cm}|p{4.5cm}|p{4.5cm}|p{4.5cm}|}
\hline
 \textbf{Task} & \bf{FORTE}& \bf{mdAlk}   & \bf{LPfeas} \\ 
\hline
 Task 1: Unconstrained 
\newline(LP bound: 32.7)& 
\#Comps: 1 (2)  Density: 32.7  AIR: 11.1
\newline 
Jiawei Han (54),
Philip S. Yu (279)
& 
\#Comps: 1 (2)  Density: 32.7  AIR: 11.1
\newline 
Jiawei Han (54),
Philip S. Yu (279)
& 
\#Comps: 1 (2)  Density: 32.7  AIR: 11.1
\newline 
Jiawei Han (54),
Philip S. Yu (279)
\\ 
\hline 
Task 2: \DB$\geq$3 
\newline(LP bound: 29.8)& 
\#Comps: 1 (3)  Density: 29.8  AIR: 7.56
\newline 
Jiawei Han (54),
Philip S. Yu (279)
\textbf{(+1)}
& 
\#Comps: 1 (3)  Density: 29.8  AIR: 7.56
\newline 
Jiawei Han (54),
Philip S. Yu (279)
\textbf{(+1)}
& 
\#Comps: 1 (3)  Density: 29.8  AIR: 7.56
\newline 
Jiawei Han (54),
Philip S. Yu (279)
\textbf{(+1)}
\\ 
\hline 
Task 3: \AI$\geq$4 
\newline(LP bound: 16.6)& 
\#Comps: 3 (1,3,3)  Density: 16.6  AIR: 10.3
\newline 
Michael I. Jordan (28),
\textit{Jiawei Han (54)},
Daphne Koller (127),
\textit{Philip S. Yu (279)},
Andrew Y. Ng (345),
Bernhard Schoelkopf (364)
\textbf{(+1)}
& 
\#Comps: 3 (1,3,3)  Density: 16.6  AIR: 10.3
\newline 
Michael I. Jordan (28),
\textit{Jiawei Han (54)},
Daphne Koller (127),
\textit{Philip S. Yu (279)},
Andrew Y. Ng (345),
Bernhard Schoelkopf (364)
\textbf{(+1)}
& 
\#Comps: 3 (1,3,3)  Density: 16.6  AIR: 10.3
\newline 
Michael I. Jordan (28),
\textit{Jiawei Han (54)},
Daphne Koller (127),
\textit{Philip S. Yu (279)},
Andrew Y. Ng (345),
Bernhard Schoelkopf (364)
\textbf{(+1)}
\\ 
\hline 
Task 4: \AI$\geq$4,
\text{$\dist_G(u,v)\leq$2},
\text{S=\{Andrew Ng\}}
\newline(LP bound: 3.91)& 
\multicolumn{2}{p{9cm}|}{
\#Comps: 1 (4)  Density: 3.89  AIR: 14.2
\newline 
Michael I. Jordan (28),
Sebastian Thrun (97),
Daphne Koller (127),
Andrew Y. Ng (345)
}
& 
\#Comps: 1 (6)  Density: 3.5  AIR: 12.5
\newline 
Michael I. Jordan (28),
Geoffrey E. Hinton (61),
Sebastian Thrun (97),
Daphne Koller (127),
Andrew Y. Ng (345),
Zoubin Ghahramani (577)
\\ 
\hline 
Task 5: \AI$\geq$4,
\text{$\dist_G(u,v)\leq$2},
\text{S=\{B.Schoelkopf\}}
\newline(LP bound: 6.11)& 
\multicolumn{3}{p{13.5cm}|}{
\#Comps: 2 (11,1)  Density: 3.54  AIR: 3.94
\newline 
\textit{Jiawei Han (54)},
\textit{Christos Faloutsos (140)},
Thomas S. Huang (146),
\textit{Philip S. Yu (279)},
\textit{Zheng Chen (308)},
Bernhard Schoelkopf (364),
\textit{Wei-Ying Ma (523)},
\textit{Ke Wang (580)}
\textbf{(+4)}
}
\\ 
\hline 
Task 6: \AI$\geq$4,
\text{$\dist_G(u,v)\leq$2},
\text{S=\{B.Schoelkopf\}},
$\sum_i c_i \leq$255 
\newline(LP bound: 2.06)& 
\multicolumn{3}{p{13.5cm}|}{
\#Comps: 1 (4)  Density: 1.24  AIR: 1.82
\newline 
Alex J. Smola (335),
Bernhard Schoelkopf (364)
\textbf{(+2)}
\medskip 
\newline 
LP+FORTE: 
\#Comps: 2 (2,2)  Density: 1.77  AIR: 2.73
\newline 
Robert E. Schapire (293),
Alex J. Smola (335),
Bernhard Schoelkopf (364),
Yoram Singer (568)
}
\\ 
\hline 
Task 7: 3$\leq$\DB$\leq$6, 
\DM$\geq$10, 
\newline(LP bound: 11.3)& 
\multicolumn{3}{p{13.5cm}|}{
\#Comps: 1 (10)  Density: 9.52  AIR: 4.96
\newline 
Haixun Wang (50),
Jiawei Han (54),
Philip S. Yu (279),
Zheng Chen (308),
Ke Wang (580)
\textbf{(+5)}
}
\\ 
\hline 
Task 8: 2$\leq$\DB$\leq$5, 
10$\leq$\DM$\leq$15, 
5$\leq$\AI$\leq$10 
\newline(LP bound: 10.7)& 
\multicolumn{3}{p{13.5cm}|}{
\#Comps: 3 (1,12,3)  Density: 7.4  AIR: 5.06
\newline 
Michael I. Jordan (28),
Jiawei Han (54),
Daphne Koller (127),
Philip S. Yu (279),
Zheng Chen (308),
Andrew Y. Ng (345),
Bernhard Schoelkopf (364),
Wei-Ying Ma (523),
Divyakant Agrawal (591)
\textbf{(+7)}
}
\\ 
\hline 
Task 9: \AI$\leq$2, 
\Th$\geq$2, 
|C|$\leq$6 
\newline (LP bound: 19)& 
\multicolumn{3}{p{13.5cm}|}{
\#Comps: 3 (2,2,2)  Density: 6.17  AIR: 1.53
\newline 
Didier Dubois (426),
Micha Sharir (447),
\textit{Divyakant Agrawal (591)},
Henri Prade (713),
Pankaj K. Agarwal (770)
\textbf{(+1)}
}
\\ 
\hline 
\end{tabular}
\label{table:Teams}
\vskip -0.2cm
\caption{Teams formed by \FORTE, \mdalk\;and \lpfeas\; for various tasks. We list the number and sizes of the found components, the (generalized) maximum density as well as the average inverse rank (AIR) based on the Citeseer list. Finally, we give name and rank of each team member with rank at most 1000. 
Experts who do not have the skill required by the task but are still included in the team are shown in \textit{italic font}.
}
\end{table*}

 In this experiment, we assess the quality of the teams obtained for several tasks with different skill requirements.
Here we consider the team formation problem (\ref{eq:teamProbSim}) in its more general setting. 
We use the generalized density objective of \eqref{eq:gen_density} where 
each vertex is given a rank $r_i$, which we define based on the number of publications of the corresponding expert. 
For each skill, we rank the experts according to the number of his/her publications in the conferences corresponding to the skill.
In this way each expert gets four different rankings; the total rank of an expert is then the minimum 
of these four ranks. 
The main advantage of such a ranking is that the experts that have higher skill are given preference, thus 
producing more competent teams. 
Note that we choose a relative measure like rank as the vertex weights instead of an absolute quantity like number of publications,
since the distribution of the number of publications varies between different fields. 
In practice such a ranking is always available and hence, in our opinion, should be incorporated.

  Furthermore, in order to identify the main area of expertise of each expert, 
  we consider his/her relative number of  publications.
   Each expert is defined to have  a skill level of 1 in skill $j$ if he has more than 25\%  
 of his/her publications in the conferences corresponding to skill $j$.
 As a distance function between authors, we use the shortest path on the \textit{unweighted version} of the DBLP graph, i.e. 
  two experts are at a distance of two, if the shortest path between the corresponding vertices in the 
 unweighted DBLP graph contains two edges. 
 Note that in general the distance function can come from other general sources beyond the input graph,
 but here we had to rely on the graph distance because of lack of other information.

   In order to assess the \textit{competence} of the found teams, we use 
   the list of the 10000 most cited authors of Citeseer \cite{Citeseer}. Note that in contrast to the skill-based ranking discussed above, this list is only used in the evaluation and \emph{not} in the construction of the graph. We compute the average inverse rank as in \cite{GajSar12} as $AIR:= 1000\cdot\sum_{i=1}^k \frac{1}{R_i}$, where $k$ is the size of the team and $R_i$ is the rank of expert $i$ on the Citeseer list of 10000 most cited authors. For authors not contained on the list we set $R_i=10001$.
 We also report the densities of the teams found in order to assess their {\it compatibility}. 
 
  We create several tasks with various constraints and compare the teams produced by \FORTE,  
 \mdalk\ and \LPFeas\ (feasible solution derived from the LP relaxation). 
 Note that 
 in our implementation we extended the \mdalk\; algorithm of $\cite{GajSar12}$ to incorporate general vertex weights, using Dinkelbach's method from fractional programming \cite{Din1967}. 
The results for these tasks are shown in Table 1. We report the upper bound given by the LP relaxation, density value, $AIR$ as well as number and sizes of the connected components. 
 Furthermore, we give the names and the Citeseer ranks of the team members who have rank at most 1000.
Note that \mdalk\; could only be applied to some of the tasks 
and \LPFeas\; failed to find a feasible team in several cases.
 
 As a first task we show the unconstrained solution where we maximize density without any constraints. 
 Note that this problem is optimally solvable in polynomial time and all methods find the optimal solution.
 The second task asks for at least three experts with skill \DB. Here 
 again all methods return the same team, which is indeed optimal since the LP bound agrees with the density of the obtained team.
 

 Next we illustrate the usefulness of the additional modeling freedom of our formulation 
 by giving an example task where obtaining meaningful, connected teams is not possible 
 with the lower bound constraints alone.
 Consider a task where we need at least four experts having skill \AI\; (Task 3).
 For this, all methods return the same disconnected team of size seven where only four members have the skill \AI.
 The other three experts possess skills \DB\ and \DM\; and are densely connected among themselves. 
 One can see from the LP bound that this team is again optimal.
  This example illustrates the major drawback of the density based objective which while 
 preferring higher density subgraphs compromises on the connectivity of the solution. 
 Our further experiments revealed that the subgraph corresponding to the skill \AI\
 is less densely connected (relative to the other skills)
 and forming coherent teams in this case is difficult without specifying additional requirements. 
 With the help of subset and distance based constraints supported by \FORTE, we can now impose  
 the team requirements more precisely and obtain meaningful teams. 
 In Task 4, we require that Andrew Y. Ng is the team leader and that all experts of the team should be within a distance of two from each other in terms of the underlying co-author graph.
  The result of our method is a densely connected and highly ranked team of 
  size four with a density of 3.89. Note that this is very close to the LP bound of 3.91.
  The feasible solution obtained by \LPFeas\ is worse than our result both in terms of density and $AIR$.
   The greedy method \mdalk\; cannot be applied to this task because of the distance constraint.
     In Task 5 we choose Bernhard Schoelkopf as the team leader while keeping the constraints from the previous task. 
     Out of the three methods, only \FORTE\ can solve this problem. It produces a large disconnected team, many members of which are highly skilled experts from the skill \DM\; and have strong connections among themselves.
     To filter these densely connected members of high expertise, we introduce a budget constraint in Task 6, where  
     we define the cost of the team as the total number of publications of its members.
     Again this task can be solved only by \FORTE\ which produces a
    compact team of four well-known \AI\ experts. 
    A slightly better solution is obtained when \FORTE\ is initialized with the infeasible solution of the LP relaxation as shown (only in this task). 
    This is an indication that on more difficult instances of \eqref{eq:teamProbSim}, 
    it pays off to run \FORTE\ with more than one starting point to get the best results. 
    The solution of the LP, possibly infeasible, is a good starting point apart from the unconstrained solution of \eqref{eq:teamProbCont}. 
    
    Tasks 7, 8 and 9 provide some additional 
    teams found by \FORTE\ for other tasks involving upper and lower bound constraints on different skills.
    As noted in Section \ref{sec:LPrel} the LP bound is loose in the presence of upper bound constraints and 
    this is also the reason why it was not possible to derive a feasible solution from the LP relaxation in these cases. 
    In fact the LP bounds for these tasks remain the same even if the upper bound constraints are dropped from these tasks. 
	
\section{Conclusions}
By incorporating various realistic constraints we have made a step forward towards a
realistic formulation of the team formation problem. 
Our method finds qualitatively better teams that are more compact and have higher densities than those found by the greedy method \cite{GajSar12}.
Our linear programming relaxation not only allows us to check the 
solution quality but also provides a good starting point for our non-convex method.
However, arguably, a potential downside of a density-based
approach is that it does not guarantee connected components.
A further extension of our approach
could aim at incorporating ``connectedness'' or a relaxed version of it as an additional
constraint.

\section*{Acknowledgements}

We gratefully acknowledge support from the Excellence Cluster MMCI at Saarland University
funded by the German Research Foundation (DFG) and the project NOLEPRO 
funded by the European Research Council (ERC).

%
%
%
\begin{appendix}
The subgradient of $S_{1}(f)$ is given by 
$s_{1}(f) = d + d^{S} + \mu_{S} I_{max}(f)$, 
where $I_{max}(f)$ is the indicator function of the largest entry of $f$.
For the subgradient of $R_2$, 
using Prop. 2.2. in \cite{Bach11}, we obtain for the subgradient
$t_{(l_j,M_j)}$ of  the terms of the form $\min\{l_j, \vol_{M_{j}}(A)\}$,
\begin{align*}
\big(t_{(l_j,M_j)}(f)\big)_i \hskip-0.1cm &=  \hskip-0.1cm \left\{\begin{array}{ll} \hskip-0.1cm 0 & \vol_{M_j}(A_{i+1}) >l_j\\
\hskip-0.1cm l_j - \vol_{M_j}(A_{i+1}) & \vol_{M_j}(A_{i}) \geq l_j, \\  &\vol_{M_j}(A_{i+1}) \leq l_j\\
	\hskip-0.1cm M_{ij} & \vol_{M_j}(A_i)<l_j	
		\end{array} \right. \ .
\end{align*}
Defining $D_{uv} := \max\{0,\ \dist(u, v) - d_{0} \}$, an element of the subgradient of the second term of $R_2$ is given as $d_D-p(f)$, where $(d_D)_i=\sum_j D_{ij}$ and  
$p(f)_i  
 \in \Big\{ \textstyle{\sum_{j=1}^m} D_{ij} u_{ij} \, |\, u_{ij}=-u_{ji}, u_{ij} \in \mathrm{sign}(f_i-f_j)\Big\},$
where $\sgn(x) := +1$, if $x >0$; -1 if $x < 0$; $[-1,1]$, if $x=0$.
In total, we obtain for the subgradient $r_2(f)$ of $R_2(f)$, 
\[
	r_{2}(f) = \gamma \textstyle{\sum_{j=1}^{p}} t_{(l_j,M_j)}(f) + \gamma \textstyle{\sum_{j=1}^{p}} t_{(k_j,M_j)}(f) +\gamma (p(f) -d_D).
\]
\end{appendix}
\medskip
\bibliographystyle{abbrv}

\end{document}